\documentclass[a4paper,onecolumn,11pt,accepted=2023-04-25]{quantumarticle}
\pdfoutput=1
\usepackage[utf8]{inputenc}
\usepackage[english]{babel}
\usepackage[T1]{fontenc}
\usepackage{amssymb}
\usepackage{amsmath}
\usepackage{hyperref}
\usepackage{enumitem}
\usepackage{subcaption}

\usepackage{tikz}
\usepackage{lipsum}

\newtheorem{theorem}{Theorem}
\newtheorem{lemma}{Lemma}

\newtheorem{remark}{Remark}

\newenvironment{proof}{{\noindent\emph{Proof:}}}{$\hfill\Box$\vspace{.1in}}

\newcommand{\ket}[1]{|#1\rangle}
\newcommand{\bra}[1]{\langle#1|}

\newcommand{\bR}{\mathbb R}
\newcommand{\bE}{\mathbb{E}}
\newcommand{\E}{\mathbb{E}}
\newcommand{\cN}{\mathcal{N}}
\newcommand{\CZ}{\textit{CZ}}
\newcommand{\x}{\hat{x}}

\begin{document}

\title{Quantum Lazy Training}

\author{Erfan Abedi}
\affiliation{QuOne Lab, Phanous Research \& Innovation Centre, Tehran, Iran}

\author{Salman Beigi}
\affiliation{QuOne Lab, Phanous Research \& Innovation Centre, Tehran, Iran}

\author{Leila Taghavi}
\affiliation{QuOne Lab, Phanous Research \& Innovation Centre, Tehran, Iran}

\maketitle

\begin{abstract}
In the training of over-parameterized model functions via gradient descent, sometimes the parameters do not change significantly and remain close to their initial values. This phenomenon is called \emph{lazy training} and motivates consideration of the linear approximation of the model function around the initial parameters. In the lazy regime, this linear approximation imitates the behavior of the parameterized function whose associated kernel, called the \emph{tangent kernel}, specifies the training performance of the model. Lazy training is known to occur in the case of (classical) neural networks with large widths. In this paper, we show that the training of \emph{geometrically local} parameterized quantum circuits enters the lazy regime for large numbers of qubits. 
More precisely, we prove bounds on the rate of changes of the parameters of such a geometrically local parameterized quantum circuit in the training process, and on the precision of the linear approximation of the associated quantum model function; both of these bounds tend to zero as the number of qubits grows. 
We support our analytic results with numerical simulations.
\end{abstract}

\section{Introduction}

The goal of achieving near-term quantum advantages have put forward quantum machine learning as one of the main applications of Noisy  Intermediate-Scale Quantum (NISQ) devices~\cite{Preskill18}. A main paradigm for achieving quantum advantage in machine learning is via quantum variational algorithms~\cite{VQA21}. In this approach, a quantum circuit consisting of parameterized gates is learned in order to fit some training data. However, this learning process through which  optimal parameters of the circuit are found, faces  challenges in practice~\cite{McClean2018, Wang2021}, and needs a thorough exploration.

Gradient descent is one of the main methods for solving optimization problems, particularly for training the parameters of a quantum circuit for machine learning. In this method, the parameters are updated by moving in the opposite direction of the gradient of a loss function to be optimized. This updating of the parameters changes not only the value of the loss function, but also the function modeled by the quantum circuit. Thus, studying the evolution of the loss and model functions during the gradient descent algorithm is crucial in understanding variational quantum algorithms.

Approximating the gradient descent algorithm with its continuous version (gradient flow)~\cite{GradientFlow}, provides us with analytical tools for the study of the evolution of a function whose parameters are optimized via gradient descent. Writing down the evolution equation for this continuous approximation, we observe the appearance of a kernel function called \emph{tangent kernel} (see Section~\ref{sec:NTK}). In short, letting $f(\Theta, x)$ be our model function with $\Theta=(\theta_1, \dots, \theta_p)$ as the parameters (weights) of the model and $x$ as a data point on which we evaluate the function, the tangent kernel is defined by
\begin{align}
    K_\Theta(x, x') = \nabla_\Theta f(\Theta, x) \cdot \nabla_\Theta f(\Theta, x').
\end{align}
Here, $\nabla_\Theta f(\Theta, x)$ is the gradient of $f(\Theta, x)$ with respect to $\Theta$ and $\nabla_\Theta f(\Theta, x)\cdot \nabla_\Theta f(\Theta, x')$ is the inner product of the gradient vector for two data points $x, x'$. The tangent kernel at some initial point $\Theta_0$ can be thought of as the kernel associated with the \emph{linear approximation} of the function given by
\begin{align}\label{eq:lin-approx-intro}
f(\Theta,x) \simeq f \big( \Theta^{(0)}, x \big) + \nabla_\Theta f \big( \Theta^{(0)}, x \big) \cdot \big( \Theta - \Theta^{(0)} \big).
\end{align}

The tangent kernel for (classical) neural networks is called the 
\emph{Neural Tangent Kernel} (NTK). It is shown in~\cite{NTK18} that although the NTK depends on $\Theta$ which varies during the gradient descent algorithm, when the \emph{width} of the neural network is large compared to its depth, the NTK remains almost unchanged. In fact, for such neural networks, the parameters $\Theta$ remain very close to their initial value $\Theta^{(0)}$. This surprising phenomenon is called \emph{lazy training}~\cite{Chizat2019lazy}. 

In the \emph{lazy regime}, since $\Theta$ is close to its initialization $\Theta^{(0)}$, the linear approximation of the function in~\eqref{eq:lin-approx-intro} is accurate. In this case, the behavior of the function under training via gradient descent follows its linear approximation, and is effectively described by the tangent kernel at initialization.  We will review these results and related concepts in more detail in Section~\ref{sec:NTK}.


\paragraph{Our results:} Our main goal in this paper is to develop 
the theory of lazy training for parameterized quantum circuits as our model function, and to generalize the results of~\cite{NTK18} to the quantum case. We prove that when the number of qubits (analogous to the width of a classical neural network) in a quantum parameterized circuit is large compared to its depth, the associated model function can be approximated by a linear model. Moreover, we show that this linear model's behavior is similar to that of the original model under the gradient descent algorithm.  

To prove the above results, we need to put some assumptions on the class of parameterized quantum circuits. The results of~\cite{NTK18} in the classical case are proven by fixing all layers of a neural network but one, and sending the number of nodes (width) in that layer to infinity. In the quantum case, assuming that we neither introduce fresh qubits nor do we measure/discard qubits in the middle of the circuit, the number of qubits is fixed in all layers. Thus, in the quantum case, unlike~\cite{NTK18}, we cannot consider layers of the circuit individually and take their width (number of qubits) to infinity independently of other layers. To circumvent this difficulty, we put some restrictions on our quantum circuits:

\begin{itemize}
\item[(i)] We assume that the circuit is geometrically local and the entangling gates are performed on neighboring qubits. For example, we assume
the qubits are arranged on a 1D or 2D lattice and the two-qubit gates are applied only on pairs of adjacent qubits. More generally, we assume that the qubits are arranged on nodes of a \emph{bounded-degree} graph and that the two-qubit gates can be applied only on pair of qubits connected by an edge. We note that this assumption arguably holds in most proposed hardware architectures of realizable quantum computers.

\item[(ii)] We also assume that 
the observable which is measured at the end of the circuit is a \emph{local operator} with its locality being in terms of the underlying bounded-degree graph mentioned above. More precisely, we assume that the observable is a sum of terms, each of which acts only on a constant number of neighboring qubits. We will offer a number of evidences to show that our results do not hold without this assumption.

\end{itemize}

Given the above assumptions, we prove the followings:
\begin{enumerate}
\item To apply the gradient descent algorithm, we usually choose the initial parameters of the circuit at random. In Theorem~\ref{thm:concentration}, We show that when choosing the initial parameters \emph{independently} at random, the quantum tangent kernel \emph{concentrates} around its average as the number of qubits tends to infinity. This means that when the number of qubits is large, at first the tangent kernel is essentially independent of the starting parameters and is fixed.  

\item We also show, in Theorem~\ref{thm:lazy}, that when the number of qubits is large, lazy training occurs; meaning that 
the parameters of the circuit do not change significantly during the gradient descent algorithm and remain almost constant. This means that the tangent kernel is fixed not only at initialization, but also during the training. As a result and as mentioned above, our model function can be approximated by a linear model which shows a behavior similar to that of the original model during the training via gradient descent. 

\end{enumerate}

These results show that in order to analyze the training behaviour of parameterized quantum circuits with the aforementioned assumptions, we may only consider the linearized model. We note that the linearized model is determined by the associated tangent kernel, which assuming that the initial parameters are chosen independently at random, is concentrated around its average. Thus, the eigenvalues of the average tangent kernel determine the training behaviour of such parameterized quantum circuits. Based on this observation, we argue in Remark~\ref{rem:exp-training} that if these eigenvalues are far from zero, then the model is trained exponentially fast. We will comment on this result in compared to the no-go results about barren plateaus in Section~\ref{sec:conclusion}.

We also provide numerical simulations to support the above results.

\paragraph{Related works:} 

The subject of tangent kernels in the quantum case has been previously studied in a few works which we briefly review.

A tangent kernel for \emph{hybrid} classical-quantum networks is considered in~\cite{quantum-enhanced-NTK21}. We note, however, that in this work the quantum part of the model is fixed and parameter-free, and only the classical part of the network is trained.

The quantum tangent kernel is considered in~\cite{QTK21} for \emph{deep} parameterized quantum circuits. In this work, a \emph{deep circuit} is a circuit with a \emph{multi-layered data encoding} which alternates between data encoding gates and parameterized unitaries.
This data encoding scheme increases the expressive power of the model function. It is shown in~\cite{QTK21} that as the number of layers increases, the changes in circuit parameters decrease during the gradient descent algorithm (a signature behavior of lazy training), and the training loss vanishes more quickly. It is also shown that the tangent kernel associated to such deep quantum parameterized circuits can outperform \emph{conventional} quantum kernels, such as those discussed in~\cite{Schuld2019} and~\cite{Vojtech2019}. We note that all of these results are based solely on numerical simulations. Moreover, the simulations are performed only for 4-qubit circuits and do not predict the behaviour of the circuits in the large width limit. 

Quantum tangent kernel of parameterized quantum circuits (for both optimization and machine learning problems) is also studied in~\cite{representation-QNTK21}. In this work, \emph{without} exploring conditions under which lazy training occurs, it is shown that in the lazy training regime (or ``frozen limit"), the loss function decays exponentially fast.

Finally, tangent kernel for \emph{quantum states} is defined in~\cite{NQS21}, and based on numerical simulations, it is shown that it can be used in the study of the training dynamics of finite-width \emph{neural network quantum states}.  

We emphasize that the missing ingredient shared by these previous works is the absence of explicit conditions on the quantum models under which the training is \emph{provably} enters the lazy regime. This missing part is addressed in our work.

\medskip

\noindent
\emph{Note added.} After publishing our work,~\cite{Liuetal2203} and~\cite{LiuLinJiang22} have also been published that further explore lazy training in quantum machine learning.

\paragraph{Outline of the paper:} The rest of this paper is organized as follows. In Section~\ref{sec:NTK}, we review the notions of tangent kernel and lazy training in more detail. In Section~\ref{sec:QNN}, we describe quantum parameterized circuits and their training. We also explain in more detail the assumption of geometric locality mentioned above, and give an explicit example of such quantum circuits. Section~\ref{sec:results} is devoted to the proof of our main results regarding quantum lazy training. In Section~\ref{sec:simulations}, we support our analytic results with numerical simulations. Concluding remarks are discussed in Section~\ref{sec:conclusion}.


\section{Tangent Kernel and Lazy Training}\label{sec:NTK}

In this section we briefly review the notion of a tangent kernel and explain the results of~\cite{NTK18} for classical neural networks.

Let $f(\Theta, x)$ be a model function which for any set of parameters $\Theta$, maps $\bR^d$ to $\mathbb{R}$. Having a training dataset $D=\{( x^{(1)}, y^{(1)} ),\ldots,( x^{(n)}, y^{(n)} )\},$ where $x^{(i)} \in \bR^d$ and $y^{(i)} \in \bR$, our goal is to find the best parameters $\Theta$ for which the outputs of our model $f(\Theta, x^{(i)})$ get close to the outputs provided in the dataset $y^{(i)}$ for all $i\in \{1, 2, \ldots, n\}$. To quantify this, we will need a metric to measure our model's ability to match our dataset. On that account, we make use of a \emph{loss function}, which in this paper is chosen to be the commonly used \emph{mean squared error} function:

\begin{equation}\label{eq:losss-func}
L(\Theta) = \frac{1}{n} \sum_{i=1}^{n} \, \frac{1}{2} \Big( f \big( \Theta, x^{(i)} \big) - y^{(i)} \Big)^2.
\end{equation}
Then, our goal is to find the optimal parameters that minimize the loss function:
\begin{equation}\label{eq:min-L}
\min_\Theta L(\Theta).
\end{equation}

We use the gradient descent algorithm to solve~\eqref{eq:min-L}. To this end, we randomly initialize parameters $\Theta = \Theta^{(0)}$ and in each step update them by moving in the opposite direction of the gradient of the loss function: $\Theta^{(t + 1)} = \Theta^{(t)} - \eta \nabla_{\Theta} L( \Theta^{(t)} )$, where $\eta$ is a fixed scalar called the \emph{learning rate} and $\nabla_{\Theta} L( \Theta^{(t)} )$ denotes the gradient of the loss function with respect to $\Theta$. This updating of parameters is repeated until a \emph{termination condition} is satisfied, e.g., the gradient vector $\nabla_{\Theta} L( \Theta^{(t)} )$ approaches zero, or the number of iterations reaches a maximum limit.

In order to analyze the gradient descent algorithm, we consider its continuous approximation. That is, we assume that the parameters are updated continuously via the gradient flow differential equation: 
$$\partial_t \Theta^{(t)} = -\nabla_{\Theta} L( \Theta^{(t)} ).$$
Then, the evolution of the model function computed at a data point $x$ is given by
\begin{align*}
\partial_t f \big( \Theta^{(t)}, x \big) &= -\nabla_{\Theta} L \big( \Theta^{(t)} \big) \cdot \nabla_{\Theta} f \big( \Theta^{(t)}, x \big)
\\ &=
-\frac{1}{n} \sum_{i=1}^{n} \Big( f \big( \Theta^{(t)}, x^{(i)} \big) - y^{(i)} \Big) \nabla_{\Theta} f \big( \Theta^{(t)}, x^{(i)} \big) \cdot \nabla_{\Theta} f \big( \Theta^{(t)}, x \big).
\end{align*}
This computation motivates the definition of the \emph{tangent kernel} as follows:
$$K_{\Theta} (x,x') = \nabla_{\Theta} f(\Theta, x) \cdot \nabla_{\Theta} f(\Theta, x').$$
We note that $K_{\Theta} (x,x')$ is a valid kernel function, since it is the inner product of two vectors. 
Then, we have
\begin{equation}\label{eq:f-der-kernel}
\partial_{t} f \big( \Theta^{(t)}, x \big) =
-\frac{1}{n} \sum_{i=1}^{n} \Big( f \big( \Theta^{(t)}, x^{(i)} \big) - y^{(i)} \Big) K_{\Theta^{(t)}} (x^{(i)}, x),
\end{equation}
The tangent kernel alone is enough to determine the evolution of the model function in the training process. 

\begin{figure}[t]
    \centering  
    \includegraphics[scale=0.8]{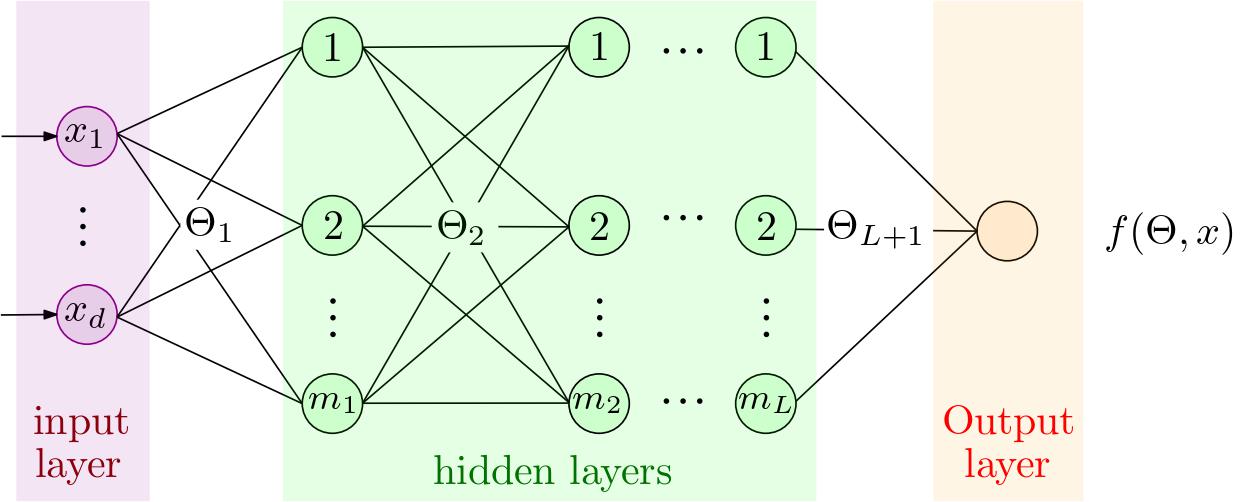}  
    \caption{ 
    A classical neural network with $L$ hidden layers. The input layer transmits a data-point $x=(x_1, \dots, x_d)\in \mathbb R^d$ to the first hidden layer through its outgoing edges like signals. Any edge of the network has a weight that transforms the passing signal.  Each node in the hidden layers applies a non-linear \emph{activation function} on its input signals and pass the result to the next layer. The output layer computes the model function $f(\Theta,x)$, where $\Theta$ denotes the set of all weights of the network. It is shown in~\cite{NTK18} that when $m_1, \dots, m_L\to \infty$ the model enters the lazy regime.
    } 
    \label{fig:CNN}
\end{figure}

Let us consider the case where $f(\Theta, x)$ comes from a neural network as in Figure~\ref{fig:CNN}. In this case, for instance, when there is only a \emph{single hidden layer}, the model function is given by
\begin{equation}\label{eq:nn-ModelFunc}
    f(\Theta, x) = \frac{1}{\sqrt m} \sum_{k=1}^{m} b_k \, \sigma \Big( \sum_{j=1}^{d} a_{kj} x_j \Big).
\end{equation}
Here, $m$ is the number of nodes in the hidden layer, $\Theta = (a_{kj} , b_{k}, 1 \leq j \leq d, 1 \leq k, \leq m)$ where $a_{kj}$ is the weight of the edge connecting $x_j$ to the $k$-th node of the hidden layer, and $b_k$ is the weight of the edge connecting the $k$-th node of the hidden layer to the output node. Moreover, $\sigma(\cdot)$ is a non-linear \emph{activation function}. Finally, following~\cite{NTK18} we introduce the normalization factor $\frac{1}{\sqrt m}$ in $f(\Theta, x)$ since we will consider the limit of this model function as $m$ tends to infinity.

When training such a neural network with a large width, i.e., large number of nodes in the hidden layers, it is observed that the initial parameters $\Theta^{(0)}$ do not change significantly, and $\Theta^{(t)}$ remains close to $\Theta^{(0)}$ until the gradient vector $\nabla_{\Theta} L( \Theta^{(t)} )$ approaches zero. This observation motivates the Taylor expansion of the model function at $\Theta^{(0)}$:
\begin{equation}\label{eq:lin-approx}
f(\Theta, x) \simeq f \big( \Theta^{(0)}, x \big) + \nabla_{\Theta} f \big( \Theta^{(0)}, x \big) \cdot \big( \Theta - \Theta^{(0)} \big).
\end{equation}
Observe that the right hand side is linear in $\Theta$ (but not in $x$). Indeed, it is a linear transformation after applying the \emph{feature map}  $x\mapsto \nabla_{\Theta} f( \Theta^{(0)}, x )$. Interestingly, the kernel function associated to this feature map is nothing but the tangent kernel $K_{\Theta^{(0)}} (x, x')$ associated to the neural network, and is called the \emph{neural tangent kernel}.

Based on the above observations, it is proven in~\cite{NTK18} that when the width of hidden layers in a neural network tends to infinity, it enters the \emph{lazy regime}, meaning that $\Theta^{(t)}$ remains close to $\Theta^{(0)}$ during the gradient descent algorithm. Moreover, it is proven that in this case, linear approximation of the model function as in~\eqref{eq:lin-approx} remains valid not only at initialization, but also during the entire training process. For more details on these results, particularly on the assumptions under which they hold, we refer to the original paper~\cite{NTK18}. We also refer to~\cite{Chizat2019lazy} for more details on lazy training.


\section{Parameterized Quantum Circuits}\label{sec:QNN}

Parameterized quantum circuits are considered as the quantum counterpart of classical neural networks~\cite{FarhiNeven2018}.
Each parameterized quantum circuit amounts to a model function and similar to neural networks, can be trained to fit some data. 

As the name suggests, a parameterized quantum circuit is a circuit with some of its gates non-fixed and dependent on some parameters. Indeed, some gates of the circuit depend on parameters denoted by $\Theta$, and some gates \emph{encode} the input $x$. 
A measurement is performed at the end of the circuit which determines the output of computation. The measurement itself could also be parameterized, but in this work, for the sake of simplicity it is assumed to be fixed. See Figure~\ref{fig:LLayer} for an example of a parameterized circuit.

Letting $U(\Theta, x)$ be the unitary associated to the circuit, and $O$ be the observable measured at the end, the resulting model function is given by
\begin{equation}\label{eq:Q-model-Func}
f(\Theta, x) = \bra{0 \cdots 0} U^{\dagger}(\Theta, x) \, O \, U(\Theta, x) \ket{0 \cdots 0}.
\end{equation}
Then, having such a model function and a dataset $D = \big\{ (x^{(1)}, y^{(1)} ), \ldots, (x^{(n)}, y^{(n)}) \big\}$, with $x^{(i)} \in \bR^d$ and $y^{(i)} \in \bR$, we may try to find the optimal $\Theta$ that  
minimizes the loss function: 
\begin{equation}\label{eq:losss-func-Q}
L(\Theta) = \frac{1}{n} \sum_{i=1}^{n} \frac{1}{2} \Big( f \big( \Theta, x^{(i)} \big) - y^{(i)} \Big)^2.
\end{equation}
To this end, as before, we initialize the parameters $\Theta$ independently at random and move towards minimizing the value of this loss function by the way of gradient descent.   

We usually arrange gates of a parameterized circuit in \emph{layers}. For instance, the circuit of Figure~\ref{fig:LLayer} consists of an encoding layer of single-qubit ($Y$-rotation) gates and $L$ layers, each of which consists of some single-qubit ($X$-rotation) gates and some two-qubit (controlled-$Z$) gates. This layer-wise structure of parameterized circuits is crucial for us since in our results, we are going to fix the number of layers $L$, and consider the limit of large number of qubits ($m\to \infty$).

In this paper, for the stability of the model, we need to assume that the parameterized gates do not change significantly by a slight change in the parameters $\Theta$. To this end, we assume that
\begin{equation}\label{eq:bound-derivative-gate}
\Big\| \frac{\partial}{\partial \theta_j} U(\Theta, x) \Big\|, \Big\|\frac{\partial^2}{\partial \theta_{i} \theta_{j}} U(\Theta, x) \Big\| \leq c, \qquad \forall i, j
\end{equation}
for some constant $c > 0$. We note that this assumption holds in most parameterized circuits in the literature, particularly when the parameterized gates are Pauli rotation (see equation~\eqref{eq:def-Rx} below).

\begin{figure}
    \centering  
    \includegraphics[scale=0.8]{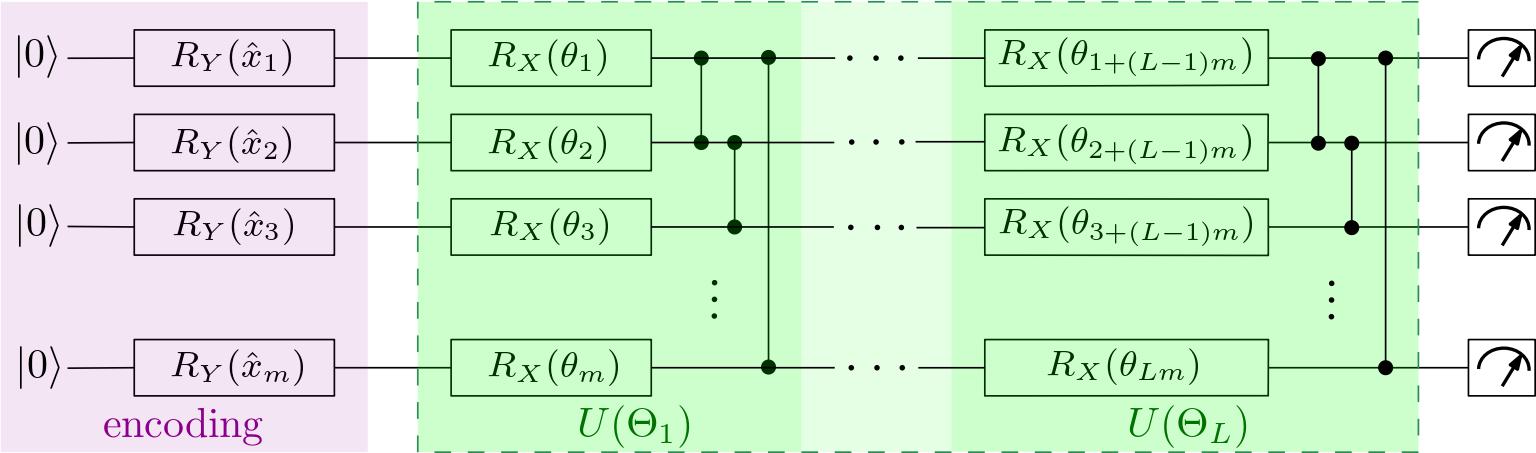}  
    \caption{
    A quantum parameterized circuit with $L$ layers of parameterized gates. Here, all the qubits are initialized at $\ket 0$, and then a layer of $Y$-rotations (i.e., $R_Y(\hat{x_{j}})= \exp{ (-i \dfrac{ \hat{ x_{j} } }{2} Y) }$ ) is applied to \emph{encode} the input $x$, where $\hat x_1, \dots, \hat x_m$ are functions (e.g., coordinates) of $x$. Next, $L$ layers of parameterized gates are applied. We assume that only the single-qubit gates are parameterized and fix the entangling gates to controlled-$Z$ gates. We assume that the qubits are arranged on a cycle, and the controlled-$Z$ in each layer are applied on all pairs of neighboring qubits.  
    }
    \label{fig:LLayer}
\end{figure}

\paragraph{Geometrically local circuits:}
As mentioned in the introduction, to prove our result we need to restrict the class of circuits to geometrically local ones. To this end, we assume that the qubits are arranged on vertices of a \emph{bounded-degree} graph (e.g., 1D or 2D lattice) and the entangling 2-qubits gates in the circuit are applied only on pairs of neighboring qubits. 
For instance, in the circuit of Figure~\ref{fig:LLayer}, we assume that the qubits are arranged on a cycle, and the controlled-$Z$ gates in each layer are applied only on pairs of neighboring qubits.

We also assume that the observable $O$ that is measured at the end of the circuit is a geometrically local one. More precisely, we assume that $O$ is given by
\begin{equation} \label{eq:observable-local}
O=\frac{1}{\sqrt m} \sum_{k=1}^{m} O_k,
\end{equation}
where $m$ is the number qubits in the circuit, and 
$O_k$ is an observable acting on the $k$-th qubit and possibly on  a constant number of qubits in its neighborhood, with $\|O_k\| \leq 1$. Moreover, as in the classical case (see, equation~\eqref{eq:nn-ModelFunc}), we introduce the normalization factor $\frac{1}{\sqrt{m}}$ in $O$ since we are considering the limit of $m \to \infty$.  
In this case, the model function~\eqref{eq:Q-model-Func} can be written as
\begin{equation}\label{eq:observable-local-MF}
    f(\Theta, x) = \frac{1}{\sqrt m} \sum_{k=1}^{m} f_{k} (\Theta, x),
\end{equation}
where
\begin{equation}\label{eq:fkx}
   f_k(\Theta, x) = \bra{0\cdots 0} U^{\dagger}(\Theta, x) \, O_k \, U(\Theta, x) \ket{0\cdots 0}.
\end{equation}

We emphasize that the assumption of geometric locality on the quantum circuit described above holds in most quantum hardware architectures. After all, the qubits in the quantum hardware should be arranged on some lattice, and usually the 2-qubits gates can only be applied on neighboring qubits. However, the assumption that the observable is geometrically local is not justified by the hardware architecture. Nevertheless, global observables usually result in \emph{barren plateaus} and a way of avoiding them is to use local observables~\cite{Localcost21}. Moreover, as
our simulations in Section~\ref{sec:simulations} show, our results do not hold for global observables. Thus, we have to somehow restrict the class of observables.

\paragraph{Example:} 
We finish this section by explaining the example of Figure~\ref{fig:LLayer} in more detail, since it will be used as our quantum circuit for simulations. 

First, we note that our data points $(x^{(i)}, y^{(i)})$ belong to $\bR^d \times \bR$, so in the circuit we need to encode each input $x$ in an $m$-qubit circuit. In the circuit of Figure~\ref{fig:LLayer} we assume that we first map $x\in \bR^d$ to some $\hat{x} \in \bR^m$ and then use the coordinates of $\hat{x}$  in the encoding layer of the circuit. The mapping $x \mapsto \hat{x}$ is arbitrary and can even be non-linear. However, for our numerical simulations we use the map:
\begin{equation}\label{eq:C-encoding}
\hat{x_{j}} = x_{j \text{ mod } d}, \quad 1 \leq j \leq m.
\end{equation}
Then, the coordinates of $\hat{x}$ are used to encode $x$ in the first layer:
\begin{equation*}
U_{\text{enc}} (x) = \prod_{j=1}^{m} R_{Y_j} (\x_j) = \prod_{j=1}^{m} \exp{ \Big( -i \frac{\x_j}{2} Y_{j} \Big) },
\end{equation*}
where $Y_j$ denotes the Pauli-$Y$ matrix acting on the $j$-th qubit.

Next, we apply $L$ parameterized unitaries $U(\Theta_{1}), \dots, U(\Theta_{L})$ where 
\begin{align*}
U(\Theta_{\ell}) = \prod_{k=1}^{m} \CZ_{k, k+1} \prod_{j=1}^{m} R_{X_j} ( \theta_{(\ell-1) m + j} ),
\end{align*}

where
\begin{equation}\label{eq:def-Rx}
R_X(\theta) = \exp{ \Big( -i \frac{\theta}{2} X \Big) },
\end{equation}
and $\CZ_{k, k+1}$ is the controlled-$Z$ gate applied on qubits $k, k+1$. Here, we assume that the qubits are arranged on a cycle, and the indices are modulo $m$.

With this specific structure for the parameterized circuit, we have
$$U(\Theta, x) = U(\Theta) U_{\text{enc}} (x) = U(\Theta_L) \cdots U(\Theta_{1}) U_{\text{enc}}(x).$$
Nevertheless, we emphasize that in this paper we do \emph{not} assume that the encoding part of the circuit is only restricted to the first layer; our results are valid even if there are gates in the middle of the circuit that encode $x$, see~\cite{datareuploading, EncodingSchuld}.

Finally, we assume that the observable is given by
$$O = \frac{1}{\sqrt m} (Z_{1} + \ldots + Z_m),$$
where $Z_{k}$ is the Pauli-$Z$ operator acting on the $k$-th qubit. Hence, the model function associated to this parameterized circuit is equal to
\begin{align}\label{eq:f-example}
f(\Theta, x) = \frac{1}{\sqrt m} \sum_{k=1}^{m} \bra{0 \cdots 0} U^{\dagger}_{\text{enc}}(x) U^\dagger(\Theta) \, Z_k \, U(\Theta) U_{\text{enc}}(x) \ket{0\cdots 0}.
\end{align}

A crucial observation which will be frequently used in our proofs is that each term in the above sum depends only on constantly many parameters (independent of $m$, the number of qubits). First, note that the last layer of controlled-$Z$ gates does not affect the model function since the controlled-$Z$ gates are diagonal in the $Z$-basis and commute with the observable. Second, and more importantly, the result of the measurement of the $k$-th qubit depends only on the \emph{light cone} of this qubit. To clarify this, let us assume that $L=2$.  In this case, the result of the measurement of the $k$-th qubit depends only on parameters $\theta_{k-1}, \theta_{k}, \theta_{k+1}, \theta_{m+k}$, see Figure~\ref{fig:lightCone}. The point is that, when $L=2$, we have 
\begin{align}
U^{\dagger}_{\text{enc}} (x) U^\dagger(\Theta)Z_k U(\Theta) U_{\text{enc}}(x) 
& = R^{\dagger}_{Y_{k-1}}(\x_{k-1}) R^{\dagger}_{Y_{k}}(\x_{k}) R^{\dagger}_{Y_{k+1}}(\x_{k+1}) \nonumber \\ 
& ~\quad R^{\dagger}_{X_{k-1}}(\theta_{k-1}) R^{\dagger}_{X_{k}}(\theta_{k}) R^{\dagger}_{X_{k+1}}(\theta_{k+1}) \nonumber\\
& ~\quad \CZ_{k-1, k} \CZ_{k, k+1} R^{\dagger}_{X_{k}}(\theta_{m + k}) \nonumber \\
& ~\quad Z_{k} \nonumber\\
& ~\quad R_{X_{k}}(\theta_{m + k}) \CZ_{k, k+1} \CZ_{k-1, k} \nonumber\\
& ~\quad R_{X_{k+1}}(\theta_{k+1}) R_{X_{k}}(\theta_{k}) R_{X_{k-1}}(\theta_{k-1}) \nonumber\\
& ~\quad R_{Y_{k+1}}(\x_{k+1}) R_{Y_{k}}(\x_{k}) R_{Y_{k-1}}(\x_{k-1}). \label{eq:Ex-Geom-Local}
\end{align}
Thus, the $f(\Theta, x)$ given by~\eqref{eq:f-example} with $L = 2$ is a sum of $m$ terms whose $k$-th term depends on $\theta_{k-1}, \theta_{k}, \theta_{k+1}$ and $\theta_{m+k}$, which together make the light cone of the $k$-th qubit (as depicted in Figure~\ref{fig:lightCone}).

\begin{figure}
    \centering  
    \includegraphics[scale=.8]{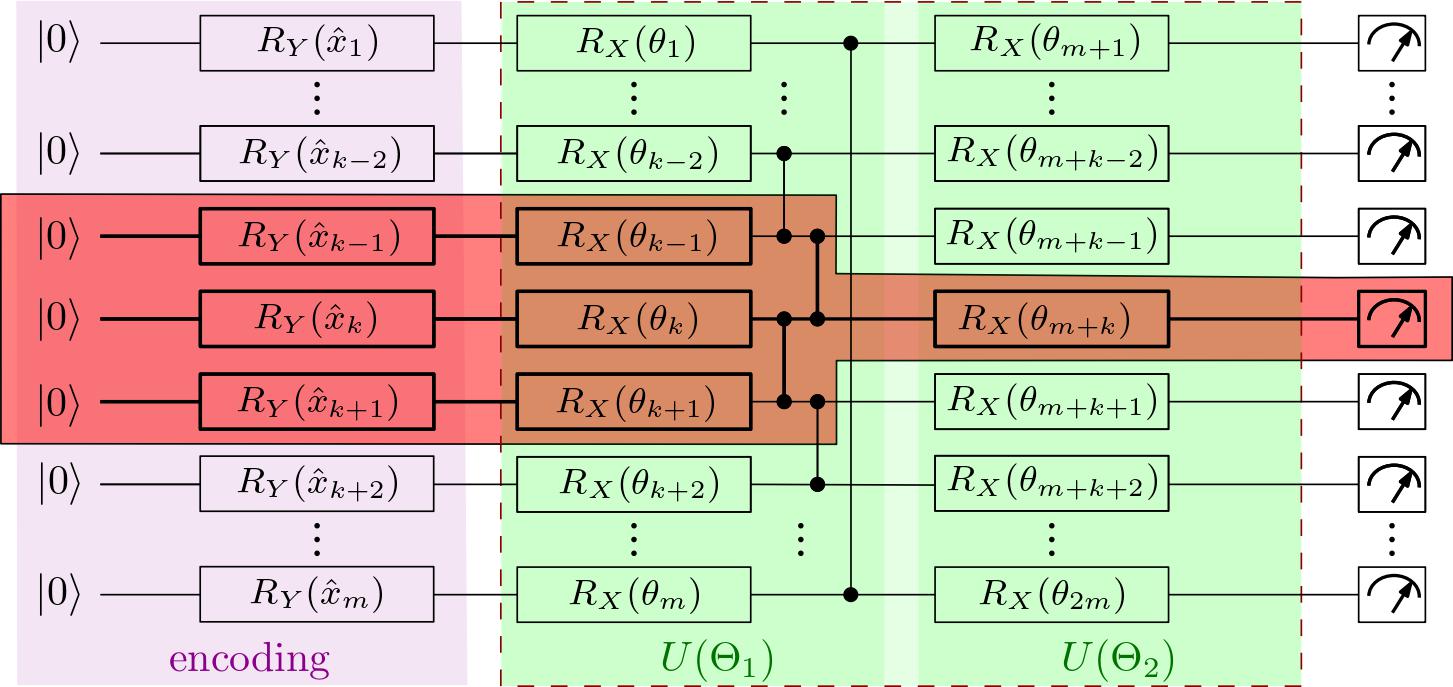}  
    \caption{
    The light cone of the $k$-th qubit of the parameterized circuit of Figure~\ref{fig:LLayer} with $L=2$ is depicted in red. This means that in order to compute the result of the $k$-th measurement $Z_{k}$, we only need to compute the red part of the circuit and ignore the rest.  We note that only the parameters $\theta_{k-1}, \, \theta_{k}, \, \theta_{k+1}$ and $\theta_{m+k}$ appear in this light cone. 
    }
    \label{fig:lightCone} 
\end{figure}

\section{Main results}\label{sec:results}

This section contain the proof of our results. We first show that under certain conditions, when the parameters are initialized independently at random, the tangent kernel is concentrated around its mean.

\begin{theorem}\label{thm:concentration}
Let $f(\Theta, x)$ be a model function associated to a \emph{geometrically local} parameterized quantum circuit on $m$ qubits as in~\eqref{eq:Q-model-Func} with $\Theta=(\theta_1, \dots, \theta_p)$ satisfying~\eqref{eq:bound-derivative-gate}. Suppose that the observable $O$ is also geometrically local given by~\eqref{eq:observable-local} where $O_{k}$ acts on the $k$-th qubit and possibly on a constant number qubits in its neighborhood, and satisfies $\|O_{k}\| \leq 1$. In this case the model function is given by~\eqref{eq:observable-local-MF} and~\eqref{eq:fkx}. Suppose that $\theta_{1}, \dots, \theta_{p}$ are chosen independently at random. Then, for any $x, x' \in \bR^d$ we have 
\begin{equation}\label{eq:concentration-result}
\Pr \Big[ \big| K_{\Theta}(x, x') - \bE[ K_{\Theta}(x,x') ] \big| \geq \epsilon \Big] \leq \exp{ \Big( -\Omega \big( \frac{m^{2} \epsilon^{2}}{pc^4} \big) \Big) }.
\end{equation}
\end{theorem} 

\medskip

\begin{remark}
We note that usually, the number of parameters in each layer of a circuit is linear in the number of qubits. Then, assuming that the number of layers $L$ is constant, $p = O(Lm) = O(m)$. In this case, the right hand side of~\eqref{eq:concentration-result} vanishes exponentially fast in $m$.
\end{remark}

As mentioned in the previous section, our main tool in proving this theorem is the geometric locality of the circuit and the observable. Based on this, following similar computations as in~\eqref{eq:Ex-Geom-Local}, we find that each term $f_k(\Theta, x)$ of the model function depends only on constantly many parameters.

In the proof of this theorem of also use McDiarmid's inequality.

\begin{lemma}[McDiarmid's Concentration Inequality~\cite{McDiarmid}] \label{lem:McDiarmid}
Let $X_{1} \ldots, X_{n}$ be independent random variables, each with values in $\mathcal{X}$. Let $f:\mathcal{X}^{n} \to \mathbb{R}$ be a mapping such that for every $i\in \{1, 2, \ldots, n\}$ and every $(x_{1}, \ldots, x_n), (x'_{1}, \ldots, x'_n) \in \mathcal{X}^{n}$ that differ only in the $i$-th coordinate (i.e., $x_{i} \neq x'_{i}$ and $\forall j \neq i: \; x_{j}=x'_{j}$),
\begin{equation*}
    | f(x_1, \ldots, x_n) - f(x'_{1}, \ldots, x'_{n}) | \leq c_{i}.
\end{equation*}
Then for any $\epsilon > 0$
\begin{equation*}
    \mathbb{P} \Big( f(X_{1}, \ldots, X_{n}) - \mathbb{E}[f(X_{1} \ldots, X_n)] \geq \epsilon \Big) \leq \exp{ \left( -\frac{2\epsilon^2}{\sum_{i=1}^{n} c_{i}^{2}} \right) }.
\end{equation*}
\hfill$\Box$
\end{lemma}

\medskip
 
\noindent
\emph{Proof of Theorem~\ref{thm:concentration}:}
Let $\cN_{k}$ be the set of indices $ j $ with $1 \leq j \leq p$ such that $U^{\dagger}(\Theta, x) \, O_k \, U(\Theta, x)$ depends on $\theta_{j}$. In other words, $\Theta_{\cN_{k}}$ is the set of $\theta_{j}$'s in the light cone of the $k$-th observable $O_{k}$. Then, we have 
\begin{align*}
f_{k}(\Theta, x) & = \bra{0 \cdots 0} U^{\dagger}(\Theta, x) \, O_k \, U(\Theta, x) \ket{0\cdots 0} \\
& = \bra{0 \cdots 0} U^{\dagger}(\Theta_{\cN_{k}}, x) \, O_k \, U(\Theta_{\cN_{k}}, x) \ket{0 \cdots 0} \\
& = f_{k}(\Theta_{\cN_{k}}, x).
\end{align*}
We note that by the assumption of geometric locality, we have $|\cN_k|=O(1)$.

Now, by the definition of the tangent kernel we have 
\begin{align} \label{eq:kxx'}
	K_{\Theta}(x, x')
	& = \nabla_{\Theta} f(\Theta, x) \cdot \nabla_{\Theta} f(\Theta, x') \nonumber \\
	& = \frac{1}{m} \sum_{k, k'=1}^{m} \sum_{j=1}^{p} \frac{\partial}{\partial \theta_{j}} f_{k}(\Theta_{\cN_{k}}, x) \cdot \frac{\partial}{\partial \theta_{j}} f_{k' (\Theta_{\cN_{k'}}, x')} \nonumber \\
	& = \frac{1}{m} \sum_{k, k'=1}^{m} \sum_{j \in \cN_{k} \cap \cN_{k'}} \frac{\partial}{\partial \theta_{j}} f_{k}(\Theta_{\cN_{k}}, x) \cdot \frac{\partial}{\partial \theta_j} f_{k'}(\Theta_{\cN_{k'}}, x'),
\end{align}
where the last equation follows since $\frac{\partial}{\partial \theta_{j}} f_{k}(\Theta_{\cN_{k}}, x) = 0$ for any $j \notin \cN_{k}$.
 
Let 
\begin{equation}\label{eq:def-set-Gamma}
\Gamma := \{(k, k', j): j \in \cN_{k} \cap \cN_{k'}\}.
\end{equation}
We note that since $O_k$ acts only on a constant number of qubits in the neighborhood of the $k$-th qubit, $\cN_{k}$ intersects $\cN_{k'}$ only if the qubits $k$ and $k'$ are geometrically close to each other (in the underlying graph). Then, since the underlying graph has a bounded degree, $\cN_k$ intersects only a constant number of $\cN_{k'}$'s. On the other hand, the size of $\cN_{k}$ is constant. Thus, for each $k$ the number of triples $(k, k', j)$ in $\Gamma$ is constant, and we have
$|\Gamma| = O(m)$.

Next, let  
$$ T_{k, k',j} = T_{k, k', j}(\Theta) = \frac{\partial}{\partial \theta_j} f_{k}(\Theta_{\cN_{k}}, x) \cdot \frac{\partial}{\partial \theta_j} f_{k'}(\Theta_{\cN_{k'}}, x').$$ 
Then,
\begin{equation*}
K_{\Theta}(x, x')=\frac{1}{m} \sum_{(k, k', j) \in \Gamma} T_{k, k', j},
\end{equation*}
can be thought of as a normalized sum of $O(m)$ terms. Note that these terms are not independent of each other; each parameter $\theta_{j}$ may appear in more than one term. Nevertheless, again by the assumption of geometric locality, each $\theta_{j}$ appears in at most constantly many terms. Therefore, by letting $\Theta, \Theta'$ be two tuples of parameters differing only at the $j$-th position (i.e., $\theta_{i} = \theta'_{i}$ for all $i \neq j$), we get
\begin{align*}
|K_{\Theta}(x, x') - K_{\Theta'}(x, x')| 
& \leq \frac{1}{m} \sum_{(k, k'):\, (k, k', j) \in \Gamma} | T_{k, k', j}(\Theta)  -  T_{k, k', j}(\Theta') | \\
& \leq \frac{1}{m} \sum_{(k, k'):\, (k, k', j) \in \Gamma} | T_{k, k', j}(\Theta)| + |T_{k, k', j}(\Theta') | \\
& = O \Big( \frac{c^2}{m} \Big),
\end{align*}
where in the last line we use~\eqref{eq:bound-derivative-gate} and the fact that for each $j$, the number of triples $(k, k', j)$ in $\Gamma$ is constant. Then, by McDiarmid's concentration inequality~\cite{McDiarmid} we have 
\begin{equation*}
\Pr\Big[ \big| K_{\Theta}(x, x') - \bE[K_{\Theta}(x, x')] \big| \geq \epsilon \Big] \leq 2 \exp{ \Big( -\frac{2 \epsilon^{2}}{p \, O({c^4} / {m^2}) } \Big) } =\exp {\Big(-\Omega \big( \frac{ m^{2} \epsilon^{2}}{pc^{4}} \big) \Big) }.
\end{equation*}
\hfill $\Box$

\medskip

The above theorem says that even though the parameters are chosen randomly at initialization, the tangent kernel is essentially fixed. This results in an essentially fixed linearized model via~\eqref{eq:lin-approx}.

The following theorem states our second main result, that the training of geometrically local quantum circuits over large number of qubits enters the lazy regime and can be approximated by a linear model.

\begin{theorem}\label{thm:lazy}
Let $f(\Theta, x)$ be a model function associated with a parameterized quantum circuit satisfying the assumptions of Theorem~\ref{thm:concentration}. Suppose that a data set $D = \big\{ (x^{(1)}, y^{(1)} ), \ldots, (x^{(n)}, y^{(n)}) \big\}$, with $x^{(i)} \in \bR^d$ and $y^{(i)} \in \bR$ is given. Assume that at initialization we choose $\Theta^{(0)} = \big(\theta_{1}^{(0)}, \dots, \theta_p^{(0)} \big)$ \emph{independently} at random, and apply the gradient flow to update the parameters in time by $\nabla \Theta^{(t)} = - \nabla_{\Theta} L(\Theta^{(t)})$, where $L(\Theta)$ is given in~\eqref{eq:losss-func-Q}. 
Then, the followings hold:
\begin{itemize}
\item[{\rm (i)}] For any $1\leq j\leq p$ we have
$$\big| \partial_{t} \theta_{j}^{(t)} \big| = O \left( \sqrt{ \frac{L \big( \Theta^{(0)} \big)}{m} } \right).$$

\item[{\rm (ii)}] For any $x, x'$ we have
$$\big|\partial_{t} K_{\Theta^{(t)} } (x, x') \big| = O\left( \sqrt{ \frac{L \big( \Theta^{(0)} \big)}{m} } \right).$$

\item[{\rm (iii)}] 
Let $\bar{f}( \bar{\Theta}, x)$ be the function associated to the linearized model, i.e.,
$$\bar f(\bar \Theta ,x) = f\big( \Theta^{(0)}, x \big) + \nabla_{\Theta} f\big( \Theta^{(0)}, x \big) \cdot \big( \bar{\Theta} - \Theta^{(0)} \big).$$
Suppose that we start with $\bar{\Theta}^{(0)} = \Theta^{(0)}$, and train the linearized model with its associated loss function denoted by $\bar{L}( \bar{\Theta}^{(t)} )$ which results in  
$$\partial_{t} \bar{f}\big( \bar{\Theta}^{(t)}, x \big)=
-\frac{1}{n} \sum_{i=1}^{n} \Big( \bar{f} \big( \bar{\Theta}^{(t)}, x^{(i)} \big) - y^{(i)} \Big) K_{\Theta^{(0)}}(x^{(i)}, x).
$$
Let 
$$\Delta(t) = \bigg(  \frac{1}{n} \sum_{i=1}^{n} \Big( f(\Theta^{(t)}, x^{(i)}) - \bar{f}( \bar{\Theta}^{(t)}, x^{(i)}) \Big)^{2} \bigg)^{1/2}.$$
Then, for all $t$ we have
\begin{align}\label{eq:bound-Delta}
\Delta(t) = O\left( \frac{ L(\Theta^{(0)}) t^{2} }{ \sqrt{m} } \right).
\end{align}

\item[{\rm (iv)}] With the notation of part (iii), for all $t$ we have
\begin{align}
    \big| L(\Theta^{(t)}) - \bar{L}( \bar{\Theta}^{(t)} ) \big| = O \left( \frac{ L(\Theta^{(0)})^{ \frac{3}{2} } t^2 }{ \sqrt{m} } \right). 
\end{align}

\end{itemize}

\end{theorem}

\medskip

Part (i) of this theorem says that parameters $\Theta^{(t)}$ do not change significantly during training. Based on this, we expect that the tangent kernel remains close to the initial tangent kernel as well. This is proven in part (ii). Next, since the tangent kernel is almost constant, we expect that our model function behaves like the linearized model in the training process (lazy training). This is formally proven in parts (iii) and (iv).

\medskip 
\begin{remark}
 The bounds of this theorem are effective when the loss function $L(\Theta^{(0)})$ at initialization is a constant independent of $m$. While we do not explore the conditions under which this holds, since $\Theta^{(0)}$ is chosen at random and $f(\Theta^{(0)}, x)$ approaches a Gaussian process, we expect to have $L(\Theta^{(0)}) = O(1)$ with high probability when we learn a \emph{bounded} function.
\end{remark} 

\medskip 
\begin{remark}\label{rem:exp-training}
Let $\bar F(t)=\big(\bar f(\bar \Theta^{(t)}, x^{(1)}), \dots, \bar{f}( \bar{\Theta}^{(t)}, x^{(n)} ) \big)$ and $Y = (y^{(1)}, \dots, y^{(n)})$. Then, since the kernel associated to the linearized model is time-independent, by~\eqref{eq:f-der-kernel} we have
$$\bar{F}(t) = \Big( \bar{F}(0) - Y \Big) e^{-\frac{t}{n}  K_{\Theta^{(0)}} } + Y.$$
This means that if $K_{\Theta^{(0)}}$ is full-rank and its minimum eigenvalue is far from zero, the training of the linearized model stops exponentially fast. In this case, the stopping time $t$ in the bounds of parts (iii) and (iv) of the theorem is small. Indeed, under the above assumption on the eigenvalues of the tangent kernel, the parameterized quantum circuit is trained exponentially fast since by part (iv) its behaviour is well-approximated by the linearized model.
\end{remark}
 
\medskip

\begin{proof}
\noindent (i) We have $\nabla \Theta^{(t)} = - \nabla_{\Theta} L(\Theta^{(t)})$, and for any $j$:
\begin{align*}
\partial_t \theta_{j}^{(t)}
= & -\frac{1}{n} \sum_{i=1}^{n} \Big( f(\Theta^{(t)}, x^{(i)}) - y^{(i)} \Big) \cdot \frac{1}{\sqrt{m}} \sum_{k=1}^{m} \frac{\partial}{\partial \theta_{j}}  f_{k}(\Theta^{(t)}, x^{(i)}) \\ 
= & -\frac{1}{n} \sum_{i=1}^{n} \Big( f(\Theta^{(t)}, x^{(i)}) - y^{(i)} \Big) \cdot \frac{1}{\sqrt{m}} \sum_{k:\; j \in \cN_{k}} \frac{\partial}{\partial \theta_{j}}  f_k(\Theta^{(t)}, x^{(i)}). 
\end{align*}
Thus, using~\eqref{eq:bound-derivative-gate} and the fact that there are a constant number of $\cN_{k}$'s containing $j$, we find that 
\begin{align*}
\big| \partial_{t} \theta_{j}^{(t)} \big|
& \leq \frac{1}{n} \sum_{i=1}^{n} \Big| f(\Theta^{(t)}, x^{(i)}) - y^{(i)} \Big| \cdot O \Big( \frac{c}{\sqrt m} \Big)\\
& \leq \Big( \frac{1}{n} \sum_{i=1}^{n} \Big| f(\Theta^{(t)}, x^{(i)}) - y^{(i)} \Big|^{2} \Big)^{1/2} \cdot O\Big( \frac{c}{\sqrt{m} } \Big)\\
& =    \Big( 2L\big( \Theta^{(t)} \big) \Big)^{1/2} \cdot O\Big(\frac{c}{\sqrt{m}} \Big).
\end{align*}
The desired bound follows once we note that we are moving in the opposite direction of the gradient of $L\big(\Theta^{(t)} \big)$ via the gradient flow equation, so $L\big( \Theta^{(t)} \big) \leq L\big( \Theta^{(0)} \big)$.

\medskip

\noindent 
(ii) Using~\eqref{eq:kxx'} we have
\begin{align*}
\partial_t K_{\Theta^{(t)}}(x,x')
& =
    \frac{1}{m} \sum_{k,k'=1}^{m} \sum_{j \in \cN_{k} \cap \cN_{k'}} \partial_{t} \bigg( \frac{\partial}{\partial \theta_{j}} f_{k}(\Theta^{(t)}, x) \cdot \frac{\partial}{\partial \theta_j} f_{k'}(\Theta^{(t)}, x') \bigg)\\ 
& = 
  \frac{1}{m} \sum_{(k, k', j) \in \Gamma} \sum_{i=1}^{p} \partial_{t} \theta_{i} \cdot \frac{\partial}{\partial \theta_{i}} \bigg( \frac{\partial}{\partial \theta_{j}} f_{k}( \Theta^{(t)}, x) \cdot \frac{\partial}{\partial \theta_j} f_{k'}(\Theta^{(t)}, x') \bigg) \\ 
& =
  \frac{1}{m} \sum_{(k,k',j) \in \Gamma} \sum_{i \in \cN_{k} \cup \cN_{k'}} \partial_{t} \theta_{i} \cdot \frac{\partial}{\partial \theta_i} \bigg( \frac{\partial}{\partial \theta_j} f_{k}(\Theta^{(t)}, x) \cdot \frac{\partial}{\partial \theta_j}f_{k'}(\Theta^{(t)}, x') \bigg).
\end{align*} 
By~\eqref{eq:bound-derivative-gate}, for any $i, j, k, k'$ we have
$$ \bigg|\frac{\partial}{\partial \theta_{i}} \bigg( \frac{\partial}{\partial \theta_j} f_{k}(\Theta^{(t)}, x) \cdot \frac{\partial}{\partial \theta_{j}} f_{k'}(\Theta^{(t)}, x') \bigg) \bigg| = O(c^{2}) = O(1).$$
\begin{align*}
   \big| \partial_{t} K_{\Theta^{(t)}}(x, x') \big| 
   = O\bigg( \frac{1}{m} \sum_{(k, k', j) \in \Gamma} \sum_{i \in \cN_{k} \cup \cN_{k'}} \big| \partial_{t} \theta_{i} \big| \bigg).
\end{align*} 
Next, recall that $|\Gamma| = O(m)$, and for any $k, k'$ the size of $\cN_{k} \cup \cN_{k'}$ is a constant. Thus, the desired bound follows from part (i).

\medskip
\noindent 
(iii) To prove this part we borrow ideas from~\cite{Chizat2019lazy}. Using~\eqref{eq:f-der-kernel} we compute
\begin{align*}
\frac{1}{2} \partial_{t} \Delta^{2}(t) 
& = \frac 1n \sum_{j=1}^{n}  \big(\partial_{t} f(\Theta^{(t)}, x^{(j)}) - \partial_{t} \bar{f}( \bar{\Theta}^{(t)}, x^{(j)}) \big) \cdot  \big( f(\Theta^{(t)}, x^{(j)}) -  \bar{f}( \bar{\Theta}^{(t)}, x^{(j)}) \big)  \\
& = -\frac {1}{n^{2}} \sum_{i, j=1}^{n}    \Big( \big(f(\Theta^{(t)}, x^{(i)}) - y^{(i)} \big) K_{\Theta^{(t)}}(x^{(i)}, x^{(j)}) \big( f(\Theta^{(t)}, x^{(j)}) -  \bar{f}( \bar{\Theta}^{(t)}, x^{(j)})\big)  \\
& \qquad \qquad \qquad  -\big( \bar{f}( \bar{\Theta}^{(t)}, x^{(i)}) - y^{(i)} \big) K_{\Theta^{(0)}}(x^{(i)}, x^{(j)}) \big( f(\Theta^{(t)}, x^{(j)}) - \bar{f}( \bar{\Theta}^{(t)}, x^{(j)} ) \big) \Big).
\end{align*}
Next, the fact that $K_{\Theta^{(0)}}$ is positive semidefinite and
$$\sum_{i, j=1}^n \big( f(\Theta^{(t)}, x^{(i)}) - \bar{f}( \bar{\Theta}^{(t)}, x^{(i)}) \big) K_{\Theta^{(0)}}(x^{(i)}, x^{(j)})  \big( f(\Theta^{(t)}, x^{(j)}) -  \bar{f} (\bar{\Theta}^{(t)}, x^{(j)}) \big) \geq 0,$$
yields,
\begin{align*}
& \frac{1}{2} \partial_{t} \Delta^{2}(t) \\
& \leq  -\frac {1}{n^2} \sum_{i, j=1}^{n} \Big( \big( f(\Theta^{(t)}, x^{(i)}) - y^{(i)} \big) K_{\Theta^{(t)}}(x^{(i)}, x^{(j)})  \big( f(\Theta^{(t)}, x^{(j)}) -  \bar{f}( \bar{\Theta}^{(t)}, x^{(j)}) \big)  \\
& \qquad \qquad \qquad -\big( f(\Theta^{(t)}, x^{(i)}) - y^{(i)} \big) K_{\Theta^{(0)}}(x^{(i)}, x^{(j)}) \big( f(\Theta^{(t)}, x^{(j)}) - \bar{f} ( \bar{\Theta}^{(t)}, x^{(j)} ) \big) \Big)\\
& = -\frac {1}{n^2} \sum_{i, j=1}^{n} \big( f(\Theta^{(t)}, x^{(i)}) - y^{(i)} \big) \cdot \big( K_{\Theta^{(t)}}(x^{(i)}, x^{(j)}) - K_{\Theta^{(0)}}(x^{(i)}, x^{(j)})\big)  \big( f(\Theta^{(t)}, x^{(j)}) - \bar{f}( \bar{\Theta}^{(t)}, x^{(j)}) \big)  \\
& \leq \frac{1}{n^2} \big\| K_{\Theta^{(t)}} - K_{\Theta^{(0)}} \big\| \cdot \bigg( \sum_{i=1}^{n} \big( f(\Theta^{(t)}, x^{(i)}) - y^{(i)} \big)^{2} \bigg)^{1/2} \cdot \bigg( \sum_{j=1}^{n} \big( f(\Theta^{(t)}, x^{(j)}) - \bar{f}( \bar{\Theta}^{(t)}, x^{(j)}) \big)^{2} \bigg)^{1/2} \\
& = \frac{\sqrt 2}{n}  \big\| K_{\Theta^{(t)}} - K_{\Theta^{(0)}} \big\| \cdot L\big( \Theta^{(t)} \big)^{1/2} \cdot \Delta(t).
\end{align*}
We also note that $\frac{1}{2} \partial_{t} \Delta^{2}(t) = \Delta(t) \cdot \partial_{t} \Delta(t)$. Therefore, 
\begin{align*}
\big| \partial_{t} \Delta(t) \big|
 \leq  
 \frac{\sqrt 2}{n} \big\| K_{\Theta^{(t)}} - K_{\Theta^{(0)}} \big\| \cdot L\big( \Theta^{(t)} \big)^{1/2} 
 \leq  
 \frac{\sqrt 2}{n} \big\| K_{\Theta^{(t)}} - K_{\Theta^{(0)}} \big\| \cdot L\big( \Theta^{(0)} \big)^{1/2}.
\end{align*}
Now, using part (ii) and the fact that $K_{\Theta^{(t)}}$ is an $n \times n$ matrix, we have 
$$\frac{1}{n}  \big\| K_{\Theta^{(t)}} - K_{\Theta^{(0)}} \big\| = O\left( \sqrt{ \frac{L\big(\Theta^{(0)}\big) }{ m}} \, t \right).$$
Therefore, 
$$\big| \partial_{t} \Delta(t) \big| = O\big( L(\Theta^{(0)} ) \frac{t}{\sqrt{m}} \big), $$
which gives the desired result by integration.

\medskip
\noindent
(iv) Using the triangle inequality for the $2$-norm, we have
\begin{align*}
    \big|L( \Theta^{(t)} ) - \bar{L}( \bar{\Theta}^{(t)}) \big| 
    & =     \Big( \sqrt{ L(\Theta^{(t)}) } + \sqrt{ \bar{L}( \bar{\Theta}^{(t)} ) } \Big) \cdot \Big| \sqrt{L(\Theta^{(t)})} - \sqrt{ \bar{L} ( \bar{\Theta}^{(t)} ) } \Big| \\
    & \leq  2 \sqrt{ L(\Theta^{(0)}) } \cdot \bigg(\frac{1}{2n} \sum_{i=1}^{n} \Big( f(\Theta^{(t)}, x^{(i)}) - \bar{f} ( \bar{\Theta}^{(t)}, x^{(i)} ) \Big)^{2} \bigg)^{1/2} \\
    & =     \sqrt{ 2L(\Theta^{(0)}) } \cdot \Delta(t) \\
    & =     O\left( \frac{L( \Theta^{(0)} )^{3/2} t^2}{\sqrt{m}} \right).
\end{align*}

\end{proof}

\section{Numerical simulations}\label{sec:simulations}

\begin{figure}
\begin{minipage}[b]{0.50\linewidth}
    \centering
    \includegraphics[width=\textwidth]{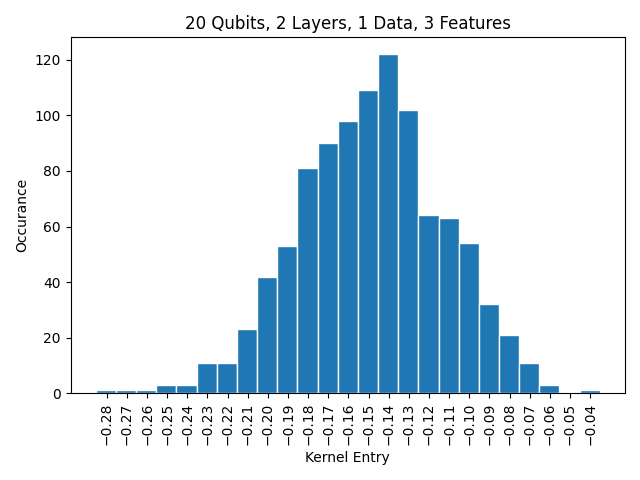}
    \caption{ 
    Histogram of $K_{\Theta}(x, x')$ for two fixed inputs $x, x'$, where $\theta_j$'s are chosen independently and uniformly at random in $[-2\pi, 2\pi]$. This histogram corresponds to the quantum circuit of Figure~\ref{fig:LLayer} for $L=2$ and $m=20$. For this experiment, the analytical mean was found to be $\bE \left[ K_{\Theta}(x,x') \right] = -0.14842$ and the empirical mean was found to be $ \bar{K}_{\Theta}(x,x') = -0.14854$. Both numbers are rounded up to the 5-th decimal point.
    }
\label{fig:histogram}
\end{minipage}
\hspace{0.4cm}
\begin{minipage}[b]{0.47\linewidth}
    \centering
    \includegraphics[width=\textwidth]{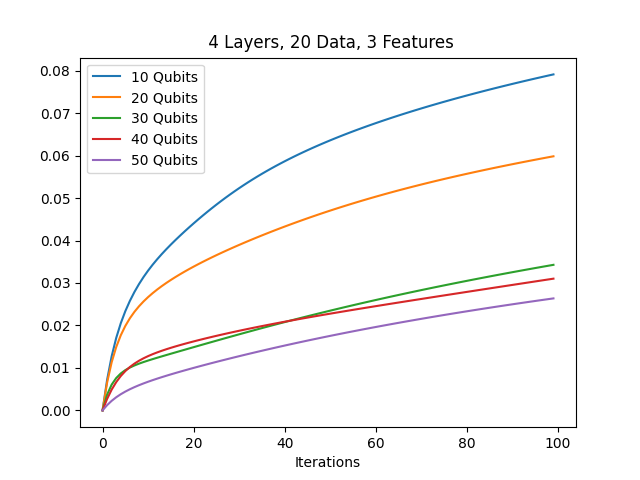}
    \caption{
    Evolution of $\frac{\| \Theta^{(t)} - \Theta^{(0)} \|}{\| \Theta^{(0)} \|}$ as a function of the number of iterations of the gradient descent algorithm. We observe that as the number of qubits $m$ increases, the relative change of parameters decrease. This means that as the number of qubits increase, training enters the lazy regime.
    \vspace{22pt}
    }
\label{fig:lazy}
\end{minipage}
\end{figure}

In this section we present numerical simulations to support  our results. To this end, we simulate the parameterized circuit of Figure~\ref{fig:LLayer}, explained in detail in Section~\ref{sec:QNN}.
To classically simulate this circuit for a large number of qubits (large $m$), we again use the idea of light cones (see Figure~\ref{fig:lightCone}). To this end, we evaluate the model function term by term, knowing that each term can be computed by a sub-circuit of constant size (when $L$ is constant). We use PennyLane~\cite{PennyLane} for our simulations.\footnote{Code is available at \href{https://github.com/phanous/quantum-lazy-training}{https://github.com/phanous/quantum-lazy-training}}

We choose the data set $D = \{ ( x^{(i)}, y^{(i)} ):  i=1, \dots, n \}$ randomly, where  $x^{(i)}$'s are  in $[-2\pi, 2\pi]$, and $y^{(i)}$'s are in $[-1, 1]$. We apply the gradient descent algorithm with a \emph{learning rate} of $\eta=1$ to train the circuit. 

We first verify Theorem~\ref{thm:concentration}. We let $L=2$, pick two random inputs $x, x'$ and compute $K_{\Theta}(x, x')$ for random choices of $\theta_{j}$ in $[-2\pi, 2\pi]$. 
Figure~\ref{fig:histogram} shows the histogram of these values. This histogram confirms that $K_{\Theta}(x, x')$ is concentrated around its average. This average is analytically computed in Appendix~\ref{app:AnalyticalComputation}, which shows
\begin{align*}
\bE\left[ K_{\Theta}(x, x') \right]
&=\frac{1}{4m} \sum_{k=1}^{m} 2 \cos(\x_{k}) \cos(\x'_{k}) +  \cos(\x_{k-1}) \cos(\x_{k}) \cos(\x_{k+1}) \cos(\x'_{k-1}) \cos(\x'_k) \cos(\x'_{k+1}).
\end{align*}

Next, in order to verify Theorem~\ref{thm:lazy}, we plot 
the relative change of the parameters $\Theta$ in the training process. That is, we plot
\begin{equation*}
    \frac{\| \Theta^{(0)} - \Theta^{(t)} \|}{\| \Theta^{(0)} \|},
\end{equation*}
where $t$ denotes the number of gradient descent iterations. As 
Figure~\ref{fig:lazy} shows, this relative change decreases by increasing the number of qubits $m$. This is an indicator of the occurrence of lazy training.

We also plot the loss functions $L(\Theta^{(t)}), \bar{L}( \Theta^{(t)} )$ of both the original quantum model and its linearized version as functions of the number of iterations in Figure~\ref{fig:LocalError}. We observe that for large numbers of qubits (e.g., $m=100$), these two loss functions have almost the same values in every step of the learning process. This confirms our results in Theorem~\ref{thm:lazy}. Moreover, we observe that as suggested in Remark~\ref{rem:exp-training}, these models converge very quickly.

\begin{figure}
\centering
\begin{subfigure}{.49 \textwidth}
    \centering
    \includegraphics[width=1\linewidth]{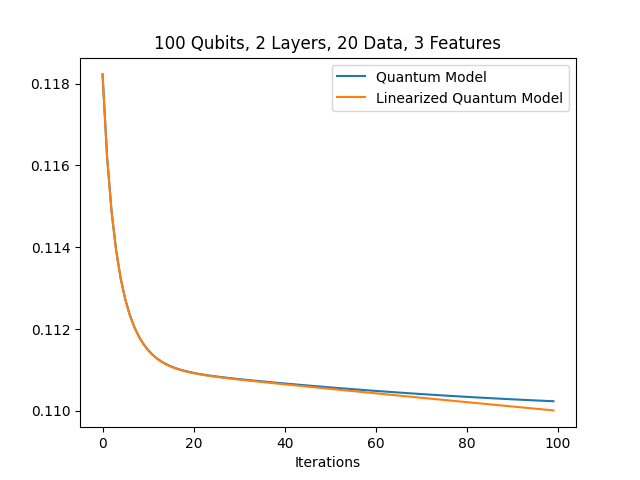}
    \caption{
        $L = 2$
    }
    \label{fig:LocalError2}
\end{subfigure}
\begin{subfigure}{.49 \textwidth}
    \centering
    \includegraphics[width=1\linewidth]{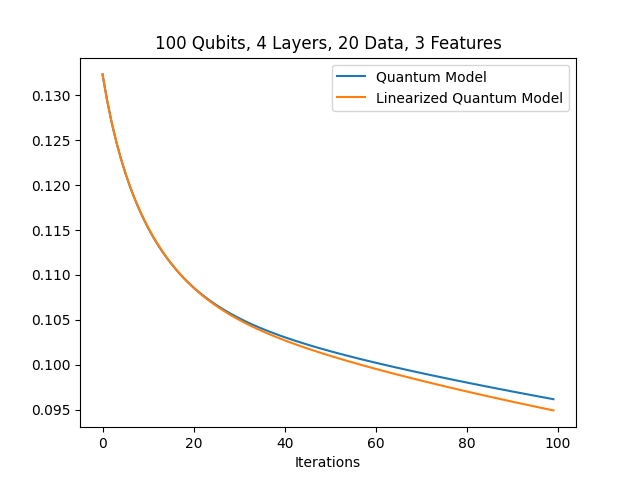} 
    \caption{
        $L = 4$
    }
    \label{fig:LocalError4}
\end{subfigure}
\caption{Evolution of the loss function as a function of the number of iterations of gradient descent. These plots correspond to the quantum circuit of  Figure~\ref{fig:LLayer} and its linearized version. We observe that for both cases of number of layers $L=2$ and $L=4$, the losses of the original quantum model and its linearized version are almost identical.}
\label{fig:LocalError}
\end{figure}

\medskip

We note that in the plots of Figure~\ref{fig:LocalError} the loss functions do not vanish. This is because, as mentioned above, the label $y^{(i)}$ for each data point $x^{(i)}$ is chosen randomly, and the quantum parameterized circuit chosen for our simulations is not expressive enough to fit such a random dataset. Alternatively, we can choose our dataset's inputs to be random $x^{(i)}$'s as before, and this time to fix the labels, pick random parameters $\Theta'$, feed the input $x^{(i)}$ to the parameterized circuit with parameters $\Theta'$, and let the outputs $y^{(i)}$ be the labels.\footnote{Of course, we forget $\Theta'$ after fixing the labels.} In this case, we make sure that our model is expressive enough to fit the dataset, and our simulations show that the loss function converges to zero as the number of iterations increase. Nevertheless, no matter how we choose the dataset, the behavior of the loss functions of the original quantum and the linearized models remain the same and they decrease with an exponential rate with the number of iterations.

\begin{figure}
\centering
\begin{subfigure}{.49\textwidth}
    \centering
    \includegraphics[width=1\linewidth]{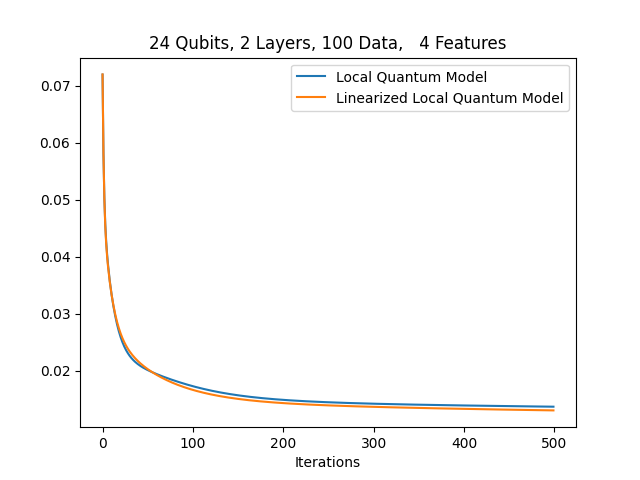}
    \caption{
        Iris flower dataset
    }
    \label{fig:iris}
\end{subfigure}
\begin{subfigure}{.49\textwidth}
    \centering
    \includegraphics[width=1\linewidth]{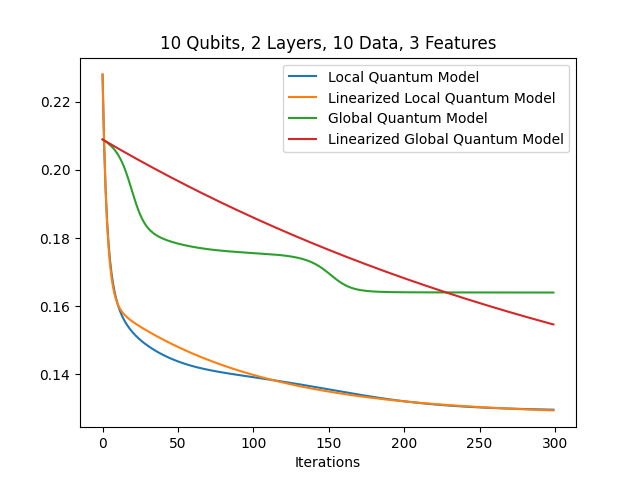}  
    \caption{
        Global model
    }
    \label{fig:Global}
\end{subfigure}
    \caption{
        (a) The loss function of the Iris flower dataset for both the quantum model and its linearized version. We see that these two loss functions are close to each other and they vanish as we increase the number of iterations. (b) The evolution of the loss function in training for the quantum circuit of Figure~\ref{fig:LLayer} and its linearized version with local and global observables. We observe that while the model with a local observable enters the lazy regime, the original and the linearized model for the global observable become separated from each other as the original global model gets trapped in a barren plateau.
    }
\label{fig:merge}
\end{figure}

We also verified our results on the Iris flower dataset.
This dataset consists of 50 data points for each three species of Iris (Iris setosa, Iris virginica and Iris versicolor), and each data point has four features. To get a binary classification problem, we picked data points corresponding to two of these three classes. We consider the same circuit as before with two layers and $m=24$ qubits. The loss function is also remained unchanged. Once again as the plot of Figure~\ref{fig:iris} shows, the loss functions of the both original quantum model and its linearized version remain close. We note that in this plot the loss function converges to zero as the number of iterations grows.

\medskip

In order to justify our assumption that the observable is geometrically local, we also consider the circuit of Figure~\ref{fig:LLayer} with a \emph{global observable}. We observe in Figure~\ref{fig:Global} that the quantum model with the global observable $O = Z_{1} Z_{2} \ldots Z_{m}$ is separated from its linearized version. This shows that the assumption of the locality of observable is necessary for lazy training. Interestingly, we also observe that the linearized version of the quantum model with a global observable doesn't learn and remains almost constant. 
This is because, as can be verified by direct computations, the associated tangent kernel is a low-rank matrix, in which case the model function has a low expressive power.

\section{Conclusion}\label{sec:conclusion}

In this paper, we proved that the training of parameterized quantum circuits that are geometrically local enters the lazy regime. This means that if the associated model function is rich enough, in which case the tangent kernel is full-rank and its eigenvalues are far from zero, training converges quickly. 

We emphasize that although in our explicit example of parameterized quantum circuit the encoding is performed only in the first layer, our results hold for general forms of data encoding including parallel and sequential ones~\cite{EncodingSchuld}. 

We proved our results under the assumptions that first, the circuit is geometrically local and second, the observable is a local operator. The first assumption is motivated by common hardware architectures, and numerical simulations suggest that the second assumption is necessary. Nevertheless, it is interesting to investigate other settings in which lazy training occurs in quantum machine learning. In particular, it would be interesting to study lazy training for quantum parameterized circuits whose number of qubits varies in different layers, i.e., fresh qubits are introduced and qubits are measured/discarded in the middle of the circuit~\cite{Beer2020training}.

Our results show that as long as the tangent kernel associated to a parameterized quantum circuit satisfying the above assumptions is full-rank and its minimum eigenvalue is far from zero, the quantum model is trained exponentially fast (see Remark~\ref{rem:exp-training}). This is in opposite direction to barren plateaus occurring in the training of certain quantum parameterized circuits~\cite{McClean2018}. The point is that the circuits considered in our work are \emph{not} random, and are geometrically local. Moreover, we consider only local observables, which remedies barren plateaus~\cite{Localcost21}.

In this paper, we fixed the loss function to be the mean squared error, yet most of the results hold for more general loss functions as well. Indeed, for a general loss function we should only modify the proof of parts (iii) and (iv) of Theorem~\ref{thm:lazy}. Modifying these parts with weaker bounds, this can be done based on ideas in~\cite{Chizat2019lazy}.

We did not explore the effect of quantum laziness in compared to its classical counterpart. For instance, how do the eigenvalues of the tangent kernel of classical and quantum models compare to each other? Which of the two models could possibly be better at generalization? We leave these questions for future works.

In the appendix, we explicitly compute the model function as well as the associated tangent kernel corresponding to a two-layer quantum circuit. We believe that such computations are insightful in understanding the expressive power of quantum parameterized circuits and their training properties.

\bibliography{References}

\bibliographystyle{unsrturl}


\pagebreak

\appendix

\section{Explicit computation of $\E[K_{\Theta}(x, x')]$}\label{app:AnalyticalComputation}

In this Appendix, we explicitly compute $\bE [K_{\Theta}(x, x')]$ for the parameterized circuit of Figure~\ref{fig:LLayer} with $L=2$ when $\theta_j$'s are chosen uniformly at random in $[-2\pi,  2\pi]$. To this end, we first explicitly compute the model function, and then compute its associated tangent kernel.

Recall that  
\begin{align*}
f(\Theta, x) =\frac{1}{\sqrt m} \sum_{k=1}^{m} f_{k}(\Theta, x),
\end{align*}
where
$$f_k(\Theta, x) = \bra{ 0 \dots 0 } \, U^{\dagger}_{\text{enc}}(x) U^{\dagger}(\Theta) Z_{k} U(\Theta) U_{\text{enc}}(x) \, \ket{0 \dots 0}.$$
We use~\eqref{eq:Ex-Geom-Local} to compute $f_{k}(\Theta, x)$.

\begin{lemma} \label{lem:compute-fk}
We have
\begin{align*}
f_k(\Theta, x)=\cos(\x_{k}) & \cos(\theta_{m+k})\cos(\theta_{k}) 
\\ &
- \cos(\x_{k-1})\cos(\x_{k})\cos(\x_{k+1})\cos(\theta_{k-1})\cos(\theta_{k+1})\sin(\theta_{k})\sin(\theta_{m+k}).
\end{align*}
\end{lemma}

\begin{proof}
Using~\eqref{eq:Ex-Geom-Local} we find that $f_k(\Theta, x)$ is the outcome of a sub-circuit acting on qubits $k-1, k, k+1$, where indices are modulo $m$. On the other hand, we have 
$$Z_k= \sum_{b_{k-1}, b_{k}, b_{k+1} \in \{0,1\}} (-1)^{b_k} \ket{b_{k-1} b_{k} b_{k+1}}\bra{b_{k-1} b_{k} b_{k+1}}.$$ 
Therefore, by~\eqref{eq:Ex-Geom-Local} we have

\begin{align*}
f_k(\Theta,  x) = & \sum_{b_{k-1}, b_{k}, b_{k+1} \in \{0,1\}} (-1)^{b_{k}} \Big| \bra{b_{k-1} b_{k} b_{k+1}} R_{X_{k}}(\theta_{m+k}) \CZ_{k, k+1} \CZ_{k-1, k} \\
& \hspace{100pt} R_X(\theta_{k-1},\theta_{k}, \theta_{k+1}) R_{Y}( \x_{k-1}, \x_{k} ,\x_{k+1} )  \ket{0 0 0} \Big|^2,
\end{align*}

where
$$R_X(\theta_{k-1}, \theta_{k}, \theta_{k+1}) = R_{X_{k+1}}(\theta_{k+1})
R_{X_{k}}(\theta_{k}) R_{X_{k-1}}(\theta_{k-1}),$$ 
and $R_{Y}(\x_{k-1}, \x_k, \x_{k+1})$ is defined similarly. Using the notation $$\ket{(\theta, x)} =  R_X(\theta) R_Y(x) \ket{0},$$
we have
$$R_X(\theta_{k-1},\theta_{k}, \theta_{k+1}) R_{Y}(\x_{k-1},\x_{k}, \x_{k+1})  \ket{0 0 0} = \ket{(\theta_{k-1}, \x_{k-1})} \ket{(\theta_{k}, \x_{k})} \ket{(\theta_{k+1}, \x_{k+1})}.$$
We also have
\begin{align*}
 &     \bra{b_{k-1}  b_{k} b_{k+1}} R_{X_k}(\theta_{m+k}) \CZ_{k, k+1}\CZ_{k-1, k} \\
= & \, \bra{b_{k-1}} \otimes \Big(\cos(\frac{\theta_{m+k}}{2}) \bra{b_{k}} - i\sin(\frac{\theta_{m+k}}{2})
\bra{b_{k+1}} \Big) \otimes \bra{b_{k+1}} \CZ_{k, k+1} \CZ_{k-1, k} \\
= & \, (-1)^{b_{k} ( b_{k+1} + b_{k-1} )} \bra{b_{k-1}} \otimes \Big( \cos(\frac{\theta_{m+k}}{2}) \bra{b_{k}} - i(-1)^{(b_{k+1}+b_{k-1})} \sin(\frac{\theta_{m+k}}{2}) \bra{b_{k+1}} \Big)\otimes \bra{b_{k+1}}\\
= & \, (-1)^{b_k(b_{k+1}+b_{k-1})} \bra{b_{k-1}} \otimes \bra{(\theta_{m+k}, b_k, b_{k+1} + b_{k-1})} \otimes \bra{b_{k+1}},
\end{align*}
where $\ket{\left( \theta, b, s \right)}$ is defined by
$$\ket{(\theta, b, s)} = \cos(\frac{\theta}{2}) \ket{b} + i(-1)^{s} \sin(\frac{\theta}{2})\ket{b + 1}.$$
Therefore,
\begin{align*}
 f_{k}(\Theta, x) = \sum_{b_{k-1}, b_{k}, b_{k+1} } (-1)^{b_{k}}  \Big|\bra{\left(\theta_{m+k}, b_k, b_{k+1} + b_{k-1} \right)} \left(\theta_{k}, \x_{k} \right) \rangle \Big|^2 \prod_{p \in \{ k-1, k+1 \}} \Big| \langle{b_{p}} \ket{ \left(\theta_{p}, \x_{p} \right) } \Big|^2 
\end{align*}
Now,  using part (ii) of Lemma~\ref{lem:phi-psi}, we obtain
\begin{align*}
f_k&(\Theta, x) \\
= & 
\sum_{b_{k-1}, b_{k+1} } \prod_{p \in \{k-1, k+1\}} \Big| \langle{b_{p}} \ket{ \left( \theta_{p}, \x_{p} \right) } \Big|^2 \cos(\x_{k}) \Big(\cos(\theta_{m+k}) \cos(\theta_k) - (-1)^{b_{k-1}+b_{k+1}} \sin(\theta_{m+k})\sin(\theta_{k})\Big).
\end{align*}
Next, using the fact that $\ket{(\theta_{p}, \x_{p})}$ is a normal vector and $\sum_{b_{p}} \big| \bra{b_p} (\theta_{p}, \x_{p}) \big|^2 = 1$, as well as part (i) of Lemma~\ref{lem:phi-psi} we find that 
\begin{align*}
f_k(\Theta, x) 
 = \cos(\x_k) 
 &   \cos(\theta_{m+k}) \cos(\theta_k) \\ 
 & - \cos( \x_{k-1} ) \cos( \x_{k} ) \cos( \x_{k+1} ) \cos( \theta_{k-1} ) \cos( \theta_{k+1} ) \sin( \theta_{k} ) \sin( \theta_{m+k} ).
\end{align*}
\end{proof}

Now recall that the tangent kernel is given by
\begin{align*}
    K_{\Theta}(x, x') 
    &= \sum_{j=1}^{2m} \frac{\partial}{\partial \theta_{j}} f(\Theta, x) \cdot \frac{\partial}{\partial \theta_{j}} f(\Theta, x') \\
    &= \frac{1}{m} \sum_{j=1}^{2m} \sum_{k, k'=1}^{m} \frac{\partial}{\partial \theta_{j}} f_{k}(\Theta, x) \cdot \frac{\partial}{\partial \theta_{j}} f_{k'}(\Theta, x').
\end{align*}
Suppose that we choose $\theta_1, \dots, \theta_{2m}$ independently and uniformly at random in $[-2\pi, 2\pi]$. We note that for such a random $\theta$ we have $\E[\cos(\theta)] = \E[\sin(\theta)] = 0$. Based on this and using Lemma~\ref{lem:compute-fk}, it is not hard to  verify that  
$$ \E \left[ \frac{\partial}{\partial \theta_{j}} f_{k}(\Theta, x) \cdot \frac{\partial}{\partial \theta_j} f_{k'}(\Theta, x') \right] = 0, \quad \forall k \neq k'.$$

Next, using the fact that $f_{k}(\Theta, x)$ depends only on parameters $\theta_{k-1}, \theta_{k}, \theta_{k+1}$ and $\theta_{m+k}$ we have
\begin{align*}
\E \left[ K_{\Theta}(x, x') \right]
& = \frac {1}{m} \sum_{k=1}^{m} \sum_{j \in \{k-1, k, k+1, m+k \}} \E \Big[\frac{\partial}{\partial \theta_j} f_{k}(\Theta, x) \cdot \frac{\partial}{\partial \theta_{j}} f_{k}(\Theta, x') \Big].
\end{align*}

Then, by Lemma~\ref{lem:fkfk'} we find that
\begin{align*}
\bE 
&   \Big[  K_{\Theta}(x, x') \Big] \\ 
& = \frac {1}{m} \sum_{k=1}^{m} \frac 24 \cos(\x_{k}) \cos(\x'_{k}) + \frac {4}{16} \cos(\x_{k-1}) \cos(\x_{k}) \cos(\x_{k+1}) \cos(\x'_{k-1}) \cos(\x'_k) \cos(\x'_{k+1}) \\
& = \frac {1}{4m} \sum_{k=1}^{m}  2 \cos(\x_{k})\cos(\x'_{k}) + \cos(\x_{k-1}) \cos(\x_{k}) \cos(\x_{k+1}) \cos(\x'_{k-1}) \cos(\x'_{k}) \cos(\x'_{k+1}).
\end{align*}

\paragraph{Auxiliary lemmas:} Here we present some auxiliary lemmas needed in the above proofs in this appendix.

\begin{lemma} \label{lem:phi-psi}
Let
$\ket{(\theta, x)} =  R_X(\theta) R_Y(x) \ket{0}$
and
$\ket{(\theta, b, s)} = \cos(\frac{\theta}{2}) \ket{b} + i(-1)^{s} \sin(\frac{\theta}{2}) \ket{b+1}$.

Then, the followings hold:
\begin{enumerate}
\item[{\rm (i)}] $ \sum_{b \in \{0, 1\}} (-1)^{b} \big| \langle b \ket{ \left(\theta, x \right)}  \big|^2 = \cos(\theta) \cos(x)$.
\item[{\rm (ii)}] $ \sum_{b \in \{0, 1\}}(-1)^{b} \big|\bra{\left(\theta, b, s \right)} \left(\theta', x \right) \rangle \big|^2 = 
\cos(x) \Big( \cos(\theta) \cos(\theta') - (-1)^{s} \sin(\theta) \sin(\theta') \Big)$.
\end{enumerate}
\end{lemma}

\begin{proof}
(i) We have
\begin{align}\label{eq:theta,x}
\ket{(\theta, x)} 
 = & R_{X}(\theta) R_{Y}(x) \ket{0} \nonumber \\ 
 = & \Big( \cos(\frac{\theta}{2})I - i\sin(\frac{\theta}{2})X \Big) \Big( \cos(\frac{x}{2})I - i\sin(\frac{x}{2})Y \Big)\ket{0} \nonumber \\
 = & \Big( \cos(\frac{\theta}{2}) \cos(\frac{x}{2}) - i \sin(\frac{\theta}{2}) \sin(\frac{x}{2}) \Big) \ket{0}
 \nonumber \\ 
   & \hspace{100pt}
 + \Big( \cos( \frac{\theta}{2} ) \sin(\frac{x}{2}) - i \sin(\frac{\theta}{2}) \cos(\frac{x}{2}) \Big) \ket{1}.
\end{align}
Therefore,
\begin{align*}
\sum_{b \in \{0,1\}} 
  & (-1)^{b} \big| \langle b \ket{\left(\theta, x\right)} \big|^2 \\ 
= & \cos^{2} \Big( \frac{\theta}{2}\Big) \cos^{2} \left( \frac{x}{2} \right) + \sin^{2} \Big( \frac{\theta}{2} \Big)\sin^2 \left( \frac{x}{2} \right) - \cos^{2} \Big( \frac{\theta}{2} \Big) \sin^{2} \left( \frac{x}{2} \right) - \sin^{2} \Big( \frac{\theta}{2} \Big) \cos^{2} \left(\frac{x}{2} \right) \\ 
= &
\bigg(\cos^{2} \Big( \frac{\theta}{2} \Big) - \sin^{2} \Big( \frac{\theta}{2} \Big) \bigg)\bigg(\cos^{2} \left( \frac{x}{2} \right) - \sin^{2} \left( \frac{x}{2} \right)\bigg) \\
= & \cos(\theta) \cos(x).
\end{align*}

\medskip
\noindent
(ii) Using~\eqref{eq:theta,x}, we have
\begin{align*}
\big| & \bra{(\theta, 0, s)}  (\theta', x) \rangle \big|^{2}\\
& = \Big| \cos(\theta/2) \left( \cos(\theta'/2) \cos(x/2) - i \sin(\theta'/2) \sin(x/2) \right)  \\
& \hspace{100pt}  -i(-1)^{s} \sin(\theta/2) \left( \cos(\theta'/2) \sin(x/2) - i \sin(\theta'/2)\cos(x/2) \right) \Big|^2 \\
& =  \Big| \cos(\theta/2) \cos(\theta'/2) \cos(x/2) - (-1)^{s} \sin(\theta/2) \sin(\theta'/2) \cos(x/2) \\
& \hspace{100pt} -i\left( \cos(\theta/2) \sin(\theta'/2) \sin(x/2) + (-1)^{s}  \sin(\theta/2)\cos(\theta'/2) \sin(x/2) \right) \Big|^2 \\
& = \cos^{2}(\theta/2) \cos^{2}(\theta'/2) \cos^{2}(x/2) + \sin^{2}(\theta/2) \sin^{2}(\theta'/2)\cos^{2}(x/2) - \frac{(-1)^{s}}{2} \sin(\theta) \sin(\theta') \cos^{2}(x/2) \\
& \quad + \cos^{2}(\theta/2) \sin^{2}(\theta'/2) \sin^2(x/2) + \sin^{2}(\theta/2) \cos^{2}(\theta'/2) \sin^{2}(x/2) + \frac{(-1)^{s}}{2} \sin(\theta) \sin(\theta') \sin^{2}(x/2),
\end{align*}
and
\begin{align*}
\big| & \bra{(\theta, 1, s)} (\theta', x) \rangle \big|^2 \\
& = \Big| \cos(\theta/2) \left( \cos(\theta'/2) \sin(x/2) - i \sin(\theta'/2) \cos(x/2) \right)  \\
& \hspace{100pt} - i(-1)^{s} \sin(\theta/2) \left( \cos(\theta'/2) \cos(x/2) - i \sin(\theta'/2)\sin(x/2) \right) \Big|^2 \\
& = \Big| \cos(\theta/2) \cos(\theta'/2) \sin(x/2) - (-1)^{s} \sin(\theta/2) \sin(\theta'/2) \sin(x/2) \\
& \hspace{100pt} -i\left( \cos(\theta/2) \sin(\theta'/2) \cos(x/2) + (-1)^{s} \sin(\theta/2)\cos(\theta'/2) \cos(x/2) \right) \Big|^2 \\
& = \cos^{2}(\theta/2) \cos^{2}(\theta'/2) \sin^{2}(x/2) + \sin^{2}(\theta/2) \sin^{2}(\theta'/2)\sin^{2}(x/2) - \frac{(-1)^{s}}{2} \sin(\theta) \sin(\theta') \sin^{2}(x/2) \\
& \quad + \cos^{2}(\theta/2) \sin^2(\theta'/2) \cos^2(x/2) + \sin^{2}(\theta/2) \cos^2(\theta'/2)\cos^{2}(x/2) + \frac{(-1)^{s}}{2} \sin(\theta) \sin(\theta') \cos^{2}(x/2).
\end{align*}
Therefore,
\begin{align*}
\big|
& \bra{(\theta, 0, s)} (\theta', x) \rangle \big|^{2} - \big| \bra{(\theta, 1, s)} (\theta', x) \rangle \big|^{2} \\
& = \Big(\cos^{2}(x/2) - \sin^{2}(x/2) \Big) \bigg(\cos^{2}(\theta/2) \cos^2(\theta'/2) + \sin^{2}(\theta/2) \sin^{2}(\theta'/2)  \\ 
& \hspace{145pt} -\frac{(-1)^{s}}{2} \sin(\theta) \sin(\theta') - \cos^{2}(\theta/2)\sin^{2}(\theta'/2) \\ 
& \hspace{145pt}
- \sin^{2}(\theta/2) \cos^{2}(\theta'/2) - \frac{(-1)^{s}}{2} \sin(\theta) \sin(\theta') \bigg) \\
& = \Big(\cos^{2}(x/2) - \sin^{2}(x/2) \Big) \bigg( \cos^{2}(\theta/2) \Big( \cos^{2}(\theta'/2) - \sin^{2}(\theta'/2) \Big)\\
& \hspace{145pt}  
-\sin^{2}(\theta/2) \Big(\cos^{2}(\theta'/2) - \sin^{2}(\theta'/2) \Big) - (-1)^{s} \sin(\theta)\sin(\theta') \bigg)\\
& =  \cos(x) \Big( \cos(\theta) \cos(\theta') - (-1)^{s} \sin(\theta) \sin(\theta') \Big).
\end{align*}

\end{proof}

\begin{lemma}\label{lem:fkfk'}
The followings hold:
 \begin{align*}
\bE \Big[ \frac{\partial}{\partial \theta_{k-1}} f_{k}(\Theta,  x) \cdot \frac{\partial}{\partial \theta_{k-1}} f_{k} 
& (\Theta,  x') \Big] =  \bE\left[ \frac{\partial}{\partial \theta_{k+1}} f_{k}(\Theta,  x) \cdot \frac{\partial}{\partial \theta_{k+1}} f_{k}(\Theta,  x') \right] \\
& = \frac {1}{16} \cos(\x_{k-1}) \cos(\x_{k}) \cos(\x_{k+1}) \cos(\x'_{k-1}) \cos(\x'_{k}) \cos(\x'_{k+1}),
\end{align*}

and 
\begin{align*}
\bE\Big[\frac{\partial}{\partial \theta_{k}} f_{k} & (\Theta, x)  \cdot \frac{\partial}{\partial \theta_{k}} f_{k}(\Theta,  x') \Big]  = \bE \left[ \frac{\partial}{\partial \theta_{m+k}} f_{k}(\Theta, x) \cdot \frac{\partial}{\partial \theta_{m+k}} f_{k}(\Theta,  x') \right]\\
& = \frac{1}{4} \cos(\x_{k}) \cos(\x'_{k}) + \frac{1}{16} \cos(\x_{k-1}) \cos(\x_{k}) \cos(\x_{k+1}) \cos(\x'_{k-1}) \cos(\x'_{k}) \cos(\x'_{k+1}).
\end{align*}

\end{lemma}

\begin{proof}
We note that
$$ \bE \left[ \cos^{2}(\theta) \right] = \bE \left[ \sin^2(\theta) \right] = \bE \Big[ \Big( \frac{\partial}{\partial \theta} \cos(\theta) \Big)^{2} \Big] = \bE \Big[ \Big( \frac{\partial}{\partial \theta} \sin(\theta) \Big)^{2} \Big] = \frac{1}{2},$$
and 
$$ \bE \left[ \sin(\theta) \cos(\theta) \right] = \bE \left[\frac{\partial}{\partial \theta} \sin(\theta) \cdot \frac{\partial}{\partial \theta} \cos(\theta) \right] = 0. $$ 
Then, using Lemma~\ref{lem:compute-fk} we compute:
\begin{align*}
   \bE \Big[ \frac{\partial}{\partial \theta_{k-1}} f_{k}(\Theta, x) \cdot 
   &   \frac{\partial}{\partial \theta_{k-1}} f_{k}(\Theta,  x') \Big] \\ 
   & = \cos(\x_{k-1}) \cos(\x_{k}) \cos(\x_{k+1}) \cos(\x'_{k-1}) \cos(\x'_{k}) \cos(\x'_{k+1}) \\
   &   \qquad \cdot \bE \left[ \Big( \frac{\partial}{\partial \theta_{k-1}} \cos(\theta_{k-1}) \Big)^{2} \cdot \cos^{2}(\theta_{k+1}) \cdot \sin^{2}(\theta_{k}) \cdot \sin^{2}(\theta_{m+k}) \right]\\
   & = \frac{1}{16} \cos(\x_{k-1}) \cos(\x_{k}) \cos(\x_{k+1}) \cos(\x'_{k-1}) \cos(\x'_{k}) \cos(\x'_{k+1}).
\end{align*}
Next, we have
\begin{align*}
    \frac{\partial}{\partial \theta_{k}} f_{k}(\Theta,  x) 
    & =\cos(\x_{k}) \cos(\theta_{m+k}) \frac{\partial}{\partial \theta_{k}} \cos(\theta_{k})\\
    & \qquad
- \cos(\x_{k-1}) \cos(\x_{k}) \cos(\x_{k+1}) \cos(\theta_{k-1}) \cos(\theta_{k+1}) \frac{\partial}{\partial \theta_{k}} \sin(\theta_{k}) \sin(\theta_{m+k}). 
\end{align*}
Therefore, since 
\begin{align*}
    \bE\Big[\frac{\partial}{\partial \theta_{k}}\cos(\theta_k) \cos(\theta_{k-1})\cos(\theta_{k+1})\frac{\partial}{\partial \theta_{k}}\sin(\theta_{k})\sin(\theta_{m+k})\Big]=0,
\end{align*}
we have
\begin{align*}
    \bE \Big[ \frac{\partial}{\partial \theta_{k}} f_{k} & (\Theta, x) \cdot \frac{\partial}{\partial \theta_{k}} f_{k}(\Theta, x') \Big] \\
    &= \cos(\x_{k}) \cos(\x'_{k}) \bE \Big[ \cos^{2}(\theta_{m+k}) \Big(\frac{\partial}{\partial \theta_{k}} \cos(\theta_{k}) \Big)^{2} \Big]\\
    & \qquad + \cos(\x_{k-1}) \cos(\x_{k}) \cos(\x_{k+1}) \cos(\x'_{k-1}) \cos(\x'_{k}) \cos(\x'_{k+1}) \\
    & \qquad \qquad \cdot \E \Big[ \cos^{2}(\theta_{k-1}) \cos^2(\theta_{k+1}) \Big( \frac{\partial}{\partial \theta_{k}} \sin(\theta_{k}) \Big)^{2} \sin^{2}(\theta_{m+k}) \Big] \\
    & = \frac{1}{4} \cos(\x_{k})\cos(\x'_{k}) +\frac{1}{16} \cos(\x_{k-1}) \cos(\x_{k}) \cos(\x_{k+1}) \cos(\x'_{k-1}) \cos(\x'_{k}) \cos(\x'_{k+1}).
\end{align*}
The other two equations hold by symmetry.
\end{proof}

\end{document}